\documentclass{article}

\usepackage[nonatbib,final]{neurips_2019}

\usepackage{hyperref} 

\usepackage[utf8]{inputenc} 
\usepackage[T1]{fontenc}    
\usepackage{url}            
\usepackage{booktabs}       
\usepackage{amsfonts}       
\usepackage{nicefrac}       
\usepackage{microtype}      
\usepackage{algorithm}
\usepackage{algpseudocode}
\usepackage{amsthm}
\usepackage{amstext}
\usepackage{xcolor}
\newtheorem{definition}{Definition}
\newtheorem{lemma}{Lemma}
\newtheorem{theorem}{Theorem}
\newtheorem{example}{Example}
\newtheorem{remark}{Remark}
\newcommand{\eps}{\epsilon}
\newcommand{\Auction}{\textsc{Auction} }

\newcommand{\pauction}{\textsc{PrivAuc} }

\newcommand{\cB}{\mathbb{\textbf{B}}}

\newcommand{\E}{\mathbb{E}}

\newcommand{\A}{\mathcal{A}}

\newcommand{\Prob}[2]{\text{Pr}_{#1}[{#2}]}

\newcommand{\sn}[1]{\iffalse \textcolor{purple}{[sn: #1]} \fi}
\newcommand{\ar}[1]{\iffalse #1 \fi}
\newcommand{\mk}[1]{\iffalse #1 \fi}
\newcommand{\ed}[1]{\iffalse #1 \fi}
\title{Optimal, Truthful, and Private Securities Lending}

\author{Emily Diana $\;$ Michael Kearns $\;$ Seth Neel $\;$ Aaron Roth}

\begin{document}

\maketitle
\begin{abstract}
\sn{should we add more details here?}
We consider a fundamental dynamic allocation problem motivated by the problem of
\emph{securities lending} in financial markets, the mechanism underlying the short selling of stocks.
A lender would like to distribute a finite number of identical copies of some scarce
resource to $n$ clients, each of whom has a private demand that is unknown to the lender.
The lender would like to maximize the usage of the resource --- avoiding allocating more to a client than her true demand --- but
is constrained to sell the resource at a pre-specified price per unit, and thus cannot use prices to incentivize truthful reporting.
We first show that the Bayesian optimal algorithm
for the one-shot problem --- which maximizes the resource's expected usage according to
the posterior expectation of demand, given reports --- actually incentivizes truthful reporting as a dominant strategy.
Because true demands in the securities lending problem are often sensitive
information that the client would like to hide from competitors, we then consider the problem under the additional desideratum
of (joint) differential privacy. We give an algorithm, based on simple dynamics for computing market equilibria,
that is simultaneously private, approximately optimal, and approximately dominant-strategy truthful. Finally, we 
leverage this private algorithm to construct an approximately truthful, optimal mechanism for the extensive form
multi-round auction where the lender does not have access to the true joint distributions between clients' requests and demands. 
\end{abstract}

\section{Introduction}
\label{intro}
In this work, we consider the allocation of a scarce commodity in
settings in which privacy concerns or demand uncertainty may be in conflict with truthful reporting.
\sn{changed to reflect new motivation}
In our model,
some number of {\em clients\/}
request desired amounts of the commodity; these requests may or may not truthfully
reflect a client's actual demand. Upon receiving the requests, a centralized {\em allocator\/}
must decide upon a distribution of the available supply and only later learns a (possibly
censored) report of the true demands.
The allocator must charge a fixed fee per unit that does not vary across clients,
and thus prices cannot be used as a tool to enforce truthfulness, as is standard in mechanism design. 

The primary motivation for this particular framework
is the problem of {\em securities lending\/} (or ``stock loan'') in financial markets~\cite{Wiki,Rothbort,Picardo}. In order to take
a short (negative) position in a publicly traded company, an investor (typically a
professional entity such as a hedge fund, asset manager or mutual fund)
must temporarily borrow shares from a party that actually owns them (typically a brokerage).
The investor then immediately sells (or ``shorts'')
the shares at the current market price; in the event of a successful short, the price subsequently declines,
and the investor buys back the shares at the lower price, pays off the stock loan (in shares),
and profits from the difference.

In this problem, the clients are the investors desiring to short a particular stock, and
the allocator is typically a brokerage firm that lends on a flat per-share fee basis~\cite{Chen,Wiki}. More
complex fee structures, including differential or volume-based pricing, are discouraged due to their complexity
and to the presence of competing flat-fee brokerages.
The commodity
is the shares to loan; it is scarce because the brokerage has a limited supply, and
demand can be high for stocks that have a great deal of ``short interest'' and are therefore
known as ``hard to borrow.''\footnote{Indeed, financial analysts often use a large number of
shares on loan as a negative indicator of market sentiment around a public company,
and hard-to-borrow tables are a valued source of such information~\cite{WSJ}.}
In general, the allocator would like to maximize the number of loaned shares actually shorted
(sold by clients in the market) since the allocator is paid fees only on the shares shorted, not on those
allocated but unused by clients. In many markets there are severe penalties to the allocator
for lending out more shares than they hold~\cite{Young-sil} (known as a ``naked short''\cite{Wiki}),
so overallocation is not an option.

Clients generally wish to be allocated their true demands. \sn{tried to add commentary about how they also might request in the morning but their true demand may not materialize until later.} But, there are natural reasons why clients may
choose not to, or be unable to, truthfully request their true demands. It is often the case that a client may request shares in the morning anticipating a trade, but their exact demand (possibly as determined by an algorithm that assesses daily market conditions) might not be decided until later in the day. Or, as with many financial market settings, clients may also have distinct privacy concerns. In
particular, sophisticated clients such as quantitative hedge funds may worry that a truthful
request for a large loan --- such as a million shares of Tesla ---  will leak the
client's view, signal, or intentions regarding Tesla to the broader market, thus compromising the client's
ability to profit from the short trade or revealing proprietary information the
client may possess.
\footnote{There is indeed evidence that short trades carry information about future prices~\cite{Lee,Sun,Boehmer}.}
It is therefore not uncommon for clients to deliberately over-request across all securities in
an effort to hide their true demands, especially if there is no penalty or disincentive for such
untruthful requests.

In this work, we formalize the setting above and show the following results:
\begin{itemize}
\item Assuming knowledge of the client's joint distributions over
true demands and requests, we derive the optimal allocation rule. This result is a generalization of the
	optimal solution for an unrelated but structurally similar allocation problem in
	financial markets known as {\em order routing in dark pools}~\cite{darkpool} (Theorem~\ref{thm:greedy1}).
\item We then show that for any client whose utility depends only
	on maximizing shares received up to her true demand, and not on any
	privacy considerations, truthfulness is a dominant strategy under the optimal allocation
	rule in a model in which clients must commit to a joint distribution on demands and reports, which the mechanism is assumed to know. Thus, such clients can safely request their true demands regardless of the
	behavior of others (Theorem~\ref{truthfulness}). 
	\item We then turn to privacy considerations. Rather than attempting to encode privacy
	explicitly into client utilities, which would result in a brittle model, we
	instead prefer the reduced-form approach of {\em joint differential privacy\/} (JDP), which provides protection against {\em any\/} privacy-related
	concerns of clients, including those permitting collusion by other clients.\footnote{See e.g. \cite{GR11} for a discussion of how the
	guarantees of differential privacy can be used to upper bound future costs resulting from information disclosure, without the need to model the specifics of the utility function.}
	We show how the optimal allocation policy can be implemented
	as a virtual ascending auction among clients, which is amenable to private implementation. The resulting algorithm provides near-optimal allocations for the
	allocator, while offering clients a solution in which truthfulness is an approximately
	dominant strategy and privacy is guaranteed (Theorems~\ref{privacc}, ~\ref{thm:privtrue}).

\item Finally we show that we can leverage our private ascending auction to construct an approximately truthful mechanism, which is approximately optimal given that clients report truthfully, for the significantly more general setting in which clients choose their reports sequentially in an arbitrarily long (or infinite) horizon extensive form game, the mechanism is not assumed to know the joint distribution on client demands and reports --- and in fact, the clients are not restricted to playing from any such static distribution. (Theorems~\ref{thm:apxtrue}, ~\ref{still:opt}).
It turns out that the privacy property of our mechanism is exactly what guarantees approximate dominant-strategy truthfulness in this more general setting, which is unusual: generally multi-round auctions cannot be shown to be dominant-strategy truthful because of the potential for bidders to threaten one another. 

\end{itemize}
\subsection{Related Work}

The work that is most closely related to our first set of results is the aforementioned paper on
order routing in dark pools~\cite{darkpool}. While motivated by a rather different trading problem,
from our perspective, their model can be viewed as the
simplification of our setting in which clients have unknown true demands but do not have the opportunity to
first provide a request or report to the allocator. The focus of~\cite{darkpool} is on the problem of
efficiently learning the optimal allocation from censored observations.
In contrast,
the presence of requests in our model makes incentive, truthfulness, and privacy
considerations most salient, which is our primary focus.

In the latter part of the paper, we build on a line of work using (variants of) differential privacy in
the context of mechanism design. Joint differential privacy --- the variant we use --- was introduced by \cite{KPRU14} as a tool for mechanism and mediator design in \emph{large games}.
The most crucially related paper is \cite{privmatchings}, which gives jointly differentially private algorithms for solving \emph{allocation problems} closely related to the one we solve in this paper.
In fact, by reformulating our problem as an allocation problem,  we can show that it fits under the setting of \cite{privmatchings}, in that our valuation functions satisfy what is called the \textit{gross substitutes} property. As such, we can use their Algorithm $3$ to compute an approximately optimal allocation subject to JDP. We give a slightly different, but technically similar private auction algorithm to theirs. Moreover, \cite{privmatchings} show how to make their allocations approximately dominant-strategy truthful \emph{if prices can be charged as a function of the allocation}. In our setting, we cannot set prices as a function of the allocation, so we need to prove dominant-strategy truthfulness of our private auction from first principles.

\section{Model}
We now formalize our allocation problem and give the necessary privacy definitions. \subsection{Definitions}
\paragraph{Basic Model: }
We model an interaction between \emph{clients} and a \emph{lender} using a \emph{mechanism} $\A$, over some fixed time horizon $T$. Let $i \in \{1,\ldots,N\}$ index clients. At each round $t$ each client $i$ has some private non-negative \emph{usage} $u_{it} \leq U$ drawn from a usage distribution $U_i$. The client submits a {\em request\/} for $r_{it}$ shares to borrow. The requests $r_{it}$ are observed directly by the lender, who will use them to choose an \emph{allocation} of shares $S_t = \{s_{it}\}$, subject to a feasibility constraint that not more than $V$ shares are offered in total: $\sum_i s_{it} \leq V$.

The client's usage $u_{it}$ is the maximum number of shares she intends to use. A client who is allocated $s_{it}$ shares obtains a payoff equal to the number of shares she actually uses: $v_i(S_t) = \min(s_{it},u_{it})$. (In fact, all of our results generalize to arbitrary bounded utility functions that are monotone increasing on $s_{it} \leq u_{it}$.) The clients may choose any arbitrary mapping from true intended usages $u_{it}$ to distributions over reports $r_{it}$, and we do not restrict how it varies round-to-round. Note that when we analyze strategic considerations, we will regard the true demand distribution $U_i$ as beyond client $i$'s control, but the choice of the mapping $u_{it} \to r_{it}$ is strategic.  By this we mean that any attempt by client $i$ to ``game" the lender comes by way of altering her requests and not by extending her usage beyond her true demand. \\

Since the request can be a randomized function of $u_{it}$, we denote the conditional request distribution by $Q_{it}(r_{it}|u_{it})$. At time $t$ define the history $H_t$ as the observed requests, allocations, and executions for each client $i$ over times $l = 1\ldots t-1$, e.g. $H_t = (r_{it}, s_{it}, v_i(S_t))_{i =1, t = 1}^{n, T}$, where again $v_i(S_t) = \min(s_{it}, u_{it})$ is the number of shares executed by client $i$ at time $t$.  Denote the set of all possible histories up to time $t$ by $\mathcal{H}_t$ and the subset of the history corresponding to client $i$ by $H_{t}^{i}$. \\

An allocation mechanism $\A$ maps the requests $r_t = (r_{it})$ at time $t$ and the history $H_t$ to allocations of shares: $\A(r_{1t},\ldots,r_{nt} ; H_t) = S_t$. An allocation rule $A$ is a one-shot algorithm that maps a set of requests $(r_{it})$ and conditional distributions  $Q_{it}(\cdot|u_{it})$ on $r_{it}$ to an allocation $S_t$. Importantly an allocation rule not only observes the requests but also has full knowledge of the conditional distributions they were drawn from. Although an allocation mechanism captures our setting of interest, allocation rules are the primary object of interest in this paper. This is because we will show that an optimal allocation mechanism has the following structure:
given $H_t$ and requests $r_t$, $\A$ estimates the conditional distributions $\hat{Q}_i(r_{it})$ as a function of $H_t$, where $\hat{Q}_i(r_{it})$ approximates the true conditional $Q_{it}(u_{it}|r_{it})$. Then treating these estimated distributions as the true conditionals, $\A$ uses an \textit{allocation rule} $A$ to compute the allocation that maximizes its expected utility with respect to those conditional distributions $A(\hat{Q}_1(r_{1t}), \ldots \hat{Q}_t(r_{nt}))$. 

We now define a \textit{strategy} for a client with respect to an allocation rule, the client's utility, the lender's utility, and joint differential privacy (JDP) of an allocation rule. We give the analogous definitions for an allocation mechanism in Section~\ref{sec:reduction}.

A \textit{strategy} for client $i$ given true demand $u_{it}$ is defined by the choice of request distribution $Q_{it}(r_{it}|u_{it})$.  Now fix a (possibly randomized) allocation rule $A$. At each round, given a set of distributions $Q_{-it}$ and reports $r_{-it}$ for the other clients and a realization of client $i$'s usage $u_{it} \sim U_i$ client $i$'s expected utility at round $t$ as a function of her own choice $Q_{it}(r_{it} | u_{it})$ is:
$$
 v^i_A(Q_{it}) = \E_{r_{it} \sim Q_{it}(\cdot|u_{it})}[v_i(A(r_{it}, r_{-it} ; Q_{it}, Q_{-it}))]
 $$
We say that $A$ is \emph{dominant-strategy truthful} if for all $i, Q_{-it}, r_{-it}, u_{it}$ client $i$'s utility function is maximized by selecting the distribution $Q_{it}(r_{it}|u_{it})$ that places all of its mass on the true demand: i.e. $Q_{it}(r_{it}|u_{it}) = 1$ if $u_{it} = r_{it}$, and $Q_{it}(r_{it}|u_{it}) = 0$ otherwise. We will denote this distribution by $\textbf{1}_{r_{it}}$.\\

At each round the \textit{lender's} realized utility for an allocation $S$ is the total number of shares executed: $\sum_i \min(s_{it},u_{it})$. Hence upon receiving reports $\{r_{it}\}$ and distributions $\{Q_{it}\}$, the lender's expected utility for a (possibly randomized) allocation rule $A$ is:

$$ v(A) = \sum_{i}\mathbb{E}_{Q_{it}, A}[\min(A(r_1,\ldots,r_n ; Q_1,\ldots,Q_n)_i, u_{it})]$$
We remark here that since the lender does not observe $u_{it}$, but does observe the draw $r_{it} \sim Q_{it}$, the expectation is taken with respect to the posterior distribution on $u_{it}$ given the request $r_{it}$. This is different than the sum of the clients' utilities, where the expectations condition on $u_{it}$ but are taken over the random draw $r_{it} \sim Q_{it}$.

\paragraph{Privacy:}
We view a \emph{dataset} as being a collection of reports $r \in \mathcal{X}^n$, where $\mathcal{X}$ is an arbitrary abstract domain (in our case, instantiated as a set of distributions $Q$ together with real valued reported demands).
\begin{definition}
Two datasets $r, r' \in \mathcal{X}^n$ are $i$-neighbors if they differ in only the report of client $i$: $r_{-i} = r'_{-i}$. We say that two datasets $r, r'$ are \emph{neighbors} if they are $i$ neighbors for any $i$.
\end{definition}

\emph{Differentially private} computations enjoy closure under post-processing, as well as composition. We defer the definition of standard differential privacy and the exact statements of these properties to the Appendix. Differential privacy is a strong guarantee that limits what an arbitrary adversary can infer about an individual \emph{even if the adversary can observe the entire output of the mechanism}. In the context of allocation problems, this guarantee is too strong. Informally, this is because a useful mechanism must be able to non-trivially vary the allocation it gives to an agent $i$ as a function of $i$'s reported demand --- see \cite{privmatchings} for a formalization of this intuition. However, allocation mechanisms have special structure, because not just their inputs, but also their outputs are  partitioned among $n$ agents. Hence, it makes sense to consider adversaries who, when trying to make inferences about agent $i$, can observe the allocation \emph{only to agents other than $i$}. This is informally what \emph{joint} differential privacy protects against.

\begin{definition}[Joint Differential Privacy \cite{KPRU14}]
A mechanism $A:\mathcal{X}^n\rightarrow \mathcal{O}^n$ is $(\epsilon, \delta)$-jointly differentially private if for every $i$, every pair of $i$-neighboring datasets $r, r'$, and for every subset $S_{-i} \subset \mathcal{O}^{n-1}$ of outputs corresponding to agents other than $i$:
$$\Pr[A(r)_{-i} \in S_{-i}] \leq \exp(\epsilon)\Pr[A(r')_{-i} \in S_{-i}] + \delta$$
If $\delta = 0$, we say $A$ satisfies $\epsilon$ joint differential privacy (JDP).
\end{definition}

Typically, to prove that an algorithm is jointly differentially private, we will first prove that a key information structure that ``coordinates'' clients in a computation is differentially private, and then apply the billboard lemma, which states that if a client's allocation is purely a function of an $\epsilon$-differentially private computation and their own private data, the overall mechanism is $\epsilon$-JDP. We defer the precise statement to the Appendix. 

In the next three sections we focus on an allocation rule at a fixed round $t$ with knowledge of the request distributions $Q_{it}$, and so we drop the subscript $t$ until Section~\ref{sec:reduction}.

\section{The Optimal Allocation Rule (Without Privacy Concerns)}
\label{sec:offline}
\subsection{Computing the Optimal Allocation}
\label{subsec:opt}
In this section, we give a simple greedy algorithm to compute the lender's optimal allocation, given knowledge of the true joint distributions $Q_i$.  In section~\ref{sec:truthful} we show that the optimal allocation rule is \emph{dominant-strategy truthful}.

Upon observing reports $r_i$ from each client, and given knowledge of $Q_i$, the lender can compute the posterior distribution $Q_i(u_i|r_i)$ on the true demand $u_i$ given $r_i$, via Bayes' rule: $$Q_i(u_i|r_i) = \frac{Q_i(r_i|u_i)U_i(u_i)}{\sum_{u'} Q(r_i|u')U_i(u')}$$
Then, we can rewrite the lender's utility more explicitly as:

\begin{equation}
\label{lenderobj}
v(S) = \sum_i \mathbb{E}_{u_i \sim Q_i(u|r_i)}[\min(s_i, u_i)]
\end{equation}
where $S = A(r_1, \ldots, r_n; Q_1, \ldots , Q_n)$. We have dropped $A$ from the expectation, because when studying the optimization problem for fixed $Q_i$, without loss of generality, we can restrict our attention to deterministic mechanisms.

We show that the simple greedy algorithm (inspired by an algorithm given in \cite{darkpool} for a different problem) presented as Algorithm \ref{optimalAllocation}, computes the allocation $S$ that exactly maximizes $v(S)$.  The algorithm operates by sequentially assigning shares $1 \ldots V$, where each share is assigned to the client $i$ most likely to be able to utilize one additional share. ``Most likely" is determined according to the posterior distribution of client demand $u_i$, conditioned on $r_i$. Given $Q_i(u|r_i)$ we denote by $T_i(s|r_i)$ the tail probability $\Prob{u_i \sim Q(u|r_i)}{u \geq s} = \sum_{s' \geq s}Q_i(s'|r_i)$.



\begin{algorithm}[h]
\caption{Greedy Allocation Rule}\label{optimalAllocation}
\begin{algorithmic}[1]
\State \textbf{Input:} $n, \{Q_i(u_i|r_i)\}_{i \in [n]}, V$
\State \textbf{Output:} feasible allocation $S = \{s_i\}$.
\Procedure{Greedy}{$n, \{Q_i(u_i|r_i)\}_{i \in [n]}, V$}
		\State Initialize $s_i = 0, \; \forall i$. \Comment{number of shares allocated to client $i$}
		\For{$t = 1 \ldots V$}
		\State Let $i^* = \text{argmax}_{i}T_i(s_i+1|r_i)$
		\State update $s_i \gets s_i + 1$
		\EndFor
		\EndProcedure
	\end{algorithmic}
\end{algorithm}

\begin{theorem}
\label{thm:greedy1}
The allocation returned by \emph{Greedy} maximizes the expected payoff for the lender: For $S$ the allocation output by greedy:
$$ S \in \arg\max_{S : \sum_i s_i = V} v(S) = \sum_i \mathbb{E}_{Q_i(u|r_i)}[\min(s_i, u_i)]$$
\end{theorem}

\subsection{Dominant-Strategy Truthfulness}
\label{sec:truthful}
We now turn our attention to the strategic question: \emph{given} that the lender is solving the allocation problem optimally for the reported $Q_i$ distributions, using Algorithm \ref{optimalAllocation}, \emph{how} will the clients behave? We show that truth telling is a dominant strategy. We note in passing that since Algorithm \ref{optimalAllocation} is also a best response for the lender, given fixed $Q_i$ distributions, this in particular means that there is a \emph{Stackelberg equilibrium} in which the clients move first, report their distributions truthfully, and then the lender optimally best responds in mechanism space.

\ar{Emily -- can you change either this or the prelims to make notation consistent? i.e. capitalization of the $q$'s, what we call the truthful strategy, etc.} \ed{Fixed}
\begin{theorem}
\label{truthfulness}
Fix a set of choices $Q_{-i}$ and reports $r_{-i}$ for all clients other than $i$, and a realization of client $i$'s usage $u_i \sim U_i$. Let $Q^T_i$ denote the truthful strategy $Q^T_i(r_i|u_i) = \textbf{1}_{r_i}$, and let $Q_i(r_i|u_i)$ denote any other strategy. Let $A$ denote the lender's optimal allocation. Then:
$$v^i_A(Q_i)  \leq v^i_A(Q^T) $$
\end{theorem}

\section{Auction Formulation}
\label{sec:game}
We first re-conceptualize the problem of computing the optimal allocation for the lender given known distributions $Q_i$ as computing the social welfare maximizing allocation with respect to a set of valuation functions for each client $i$. Using this formulation of the problem, we give an algorithm (based on \cite{KelsoCraw}) that uses an ascending price auction formulation to compute an approximately optimal allocation. This formulation will naturally lend itself to computing the allocation in a way that satisfies (joint) differential privacy, using techniques similar to those in \cite{privmatchings}.

\subsection{Optimal Allocation}

Consider a setting in which $V$ identical units of a good are being sold to $n$ bidders, each of whom has an arbitrary decreasing marginal valuation function for up to $U$ units of each good. We model bidders as having quasi-linear utility for money and wish to find the welfare maximizing allocation. We can map our problem onto this setting as follows: For each agent $i$ who requests $r_i$ shares and has a posterior demand distribution $Q_i(u_i | r_i)$, we define the valuation function for agent $i$ as a function of the number of units of the good they receive, as follows:
 $$v_i(s) = \frac{1}{U}\E_{u \sim Q_i(u_i |r_i)}\min(s, u_i) = \sum_{j = 1}^{s}\Prob{u \sim Q_i(u_i|r_i)}{u \geq j}$$
 
Given an allocation $S = (s_1,\ldots,s_n)$ of $V$ shares to $n$ clients, we define the total social welfare to be $v(S) = \frac{1}{V}\sum_iv_i(s_i)$. It is immediate that an allocation that maximizes the social welfare is the optimal allocation from the perspective of the lender in our problem: it maximizes the expected number of shares executed with respect to the posterior distributions on true demand. Since our new ascending price auction will compute the approximately optimal allocation for a wider range of allocation problems than the securities lending problem that is our main interest, we abstract away the securities lending setting and state our results in full generality.

Suppose that we are in the general setting of allocating $V$ identical copies of a good to $n$ bidders, each of whom possess a valuation function $v_i: [U] \to [0,1]$, with the diminishing marginal returns (DMR) property, defined below. 

\begin{definition}
A valuation function $v_i$ is said to have the diminishing marginal returns property (DMR) if for all $s \leq j \leq U-1$:
$$
v_i(s+1)-v_i(s) \geq v_i(j+1)-v_i(j)
$$
\end{definition}

Then the algorithm \Auction (see Algorithm \ref{auction1} in the Appendix) efficiently computes an approximately social welfare maximizing allocation for any such problem. 

\begin{theorem}
\label{thm:auction1}
$\Auction(V, \alpha, U)$ terminates after at most $\frac{V}{\alpha} + 1$ rounds. At termination, $S$ constitutes an $\frac{\alpha V}{n}$-optimal allocation:
$$ v(S) \geq \max_{S'}v(S') -\frac{\alpha V}{n}$$
\end{theorem}

\subsection{Private Auction}
\label{sec:auction}
We modify \Auction so that it will guarantee joint differential privacy, following the approach of \cite{privmatchings}. We will show that for sufficiently large auctions,  e.g. $n$ sufficiently large relative to $V, \epsilon$, we can achieve privacy while still outputting a high-quality allocation. Finally, we will show that our private auction remains approximately dominant-strategy truthful. We give full pseudocode for \pauction in the Appendix.

At a high level, we modify \Auction in a few ways in order to make it jointly differentially private:
\begin{enumerate}
\item The running count $T_B$ of the total number of bids placed so far will be computed approximately using a differentially private estimator. Since the price at each round is computed purely as a function of $T_B$, the price trajectory will be differentially private as well.
\item Rather than terminating when $\cB = 0$, the algorithm will terminate when $\cB < \rho n$ (early stopping). This will serve to limit the maximum number of times any single buyer can place a bid, which will aid us in bounding the error of the differentially private bid count.
\item Rather than running the auction with a supply of $V$ shares, we will run the auction with a supply of $V-E$ shares, where $E$ corresponds to the maximum error of our differentially private bid counter; this ensures that our computed allocation (which now may over or under allocate with respect to its target supply) is always feasible.

\end{enumerate}
Through the run of \pauction, each player computes her own allocation purely as a function of the (private) trajectory of prices. As a result, the entire procedure will satisfy joint differential privacy by Lemma~\ref{lem:billboard}.

We define the lender's utility under the optimal (non-private) allocation of $V$ shares by $OPT_V = \max_{\{S: \sum s_i = V\}}v(S) = \max_{\{S: \sum s_i = V\}} \frac{1}{V}\sum_iv_i(s_i)$. In the theorem statement to follow, $\rho$ governs the accuracy of the private allocation and is an input to $\pauction$ which defines the early stopping criterion. Any fixed value of $\rho$ specifies a range of instance parameters $(n, V, \alpha)$ for which the accuracy theorem holds.

\begin{theorem}
\label{privacc}
\begin{enumerate}
\item $ \alpha (\frac{V}{\rho})\leq n \leq \frac{1}{\alpha}(\frac{V}{\rho})$
\item $n = \Omega(\sqrt{\frac{\log(1/\beta)\log(2V/\alpha\rho)^{5/2} V}{\epsilon \alpha \rho^2}})$ $\;  \Leftrightarrow n \geq 8E/\rho$
\item $n = \Omega(\frac{\log(1/\beta)\log(V/\alpha\rho)^{5/2}}{\epsilon \alpha \rho^2})$ $\; \Leftrightarrow \frac{E}{V} \leq \rho/8$
\end{enumerate}

For $\alpha, \beta, \rho, n, V, \epsilon$ such that $(1), (2), (3)$ hold, $\pauction$ satisfies $(\epsilon,\beta)$-JDP, and if $S$ is the allocation returned by $\pauction$, with probability $1-\beta$:
$$v(S) \geq  (1-\rho)OPT_V - \rho$$ 

\end{theorem}

Theorem~\ref{privacc} tells us that for $\alpha$ sufficiently small, and for $n = \Omega(\frac{\sqrt{V}}{\epsilon \rho})$, we are able to achieve $(\epsilon,\beta)$-JDP and near-optimal welfare.

\pauction still incentivizes truthful reporting as an approximate dominant strategy for almost all agents. As $n$ grows, both approximations become perfect. Our proof of this crucially uses the privacy of $\pauction$.

\begin{theorem}
\label{thm:privtrue}
\ar{Make notation consistent with prelims}
Let client $i$ have true demand $u_i \sim Q_i(u_i)$. Suppose she selects an arbitrary conditional distribution $Q_i(r|u_i)$ and draws $r_i \sim Q_i(r|u_i)$, which in turn induces the lender's posterior $Q_i(u|r_i)$. Let $Q^T$ denote the truthful conditional distribution $Q^T(r|u_i) = \textbf{1}\{r = u_i\}$. Let $A$ denote $\pauction$. Then, for at least $(1-\sqrt{\beta + (1-\beta)\rho})n$ clients $i$, \ar{Any way to set $\rho$ so it doesn't show up in these bounds? Ideally $\rho$ shrinks with some parameter of the market, like $n$}
$$ v^i_A(Q^T) \geq e^{-\epsilon}v^i_A(Q_i)-e^{-\epsilon}\frac{\sqrt{\beta + (1-\beta)\rho}}{1-\sqrt{\beta + (1-\beta)\rho}}$$
\end{theorem}

\section{An Approximately Optimal Allocation Mechanism}
\label{sec:reduction}
\subsection{Setting}
Finally we show how $\pauction$ can be leveraged to give an approximately optimal mechanism in the general setting where at each round each client $i$ has the freedom to (adaptively) choose an arbitrary mapping $L_{i}^t: \mathcal{H}^{i}_t \times [U] \to [U]$ that maps the realized history and demand $H_t, u_{it}$ respectively, to a request $r_{it}$. We first give the relevant definitions for allocation mechanisms.

We define a strategy for client $i$  as a set of randomized mappings $L_{t}^{i}: \mathcal{H}_t^{i} \times [U] \to [U]$ for $t = 1 \ldots T$, that map the observed history for that client $H_{t}^{i}$ and the demand $u_{it}$ at round $t$, to the request $r_{it}$.

Given an allocation mechanism $\A$, and a client $i$, the utility of client $i$ is defined as:
$$
v_{\A}^{i}(L_{1}^i, \ldots, L_{T}^{i}) = \sum_{t = 1}^{T}\E[v_i(\A(r_{it}, r_{-it} ; H_t)],
$$
where the expectation is taken over the randomness in $\A$ and over the $L_{i}^{t}$ for all $i,t$. We say that $\A$ is \emph{dominant-strategy truthful} if for all $i$, fixing the strategies at each round of all other clients $(L_{t}^{-i})_{t=1}^{T}$,  client $i$'s utility function is maximized by the truthful strategies $L_{t}^{i*}$ defined by: \\
$L_{t}^{i*}(H_t, u_{it}) = u_{it}$ for all $u_{it} \in [U], H_t^{i} \in \mathcal{H}_t^{i}, \forall t \in [T]$.

The lender's utility for an allocation mechanism is defined similarly to summing up the lender's utility for an allocation rule used at each round. However, since an allocation mechanism does not know the true conditionals $Q_{it}(u_{it}|r_{it})$ at each round and only observes the requests $r_{it}$, the utility for a mechanism is taken in expectation over the conditional $u_{it}|r_{it}, H_t$, rather than the posterior $Q_{it}(u_{it}|r_{it})$: 
$$ v(A) = \sum_{t} \sum_{i}\mathbb{E}_{u_{it} \sim Q_{it}(u_{it}| r_{it}, H_t), A}[\min(A(r_1,\ldots,r_n ; Q_1,\ldots,Q_n)_i, u_{it})]$$

\subsection{An Optimal Mechanism}
Consider \emph{Greedy Mechanism}  and the \emph{Private Greedy Mechanism} defined in the Appendix. Both mechanisms estimate the posterior distribution $Q_{it}(u_{it}|r_{it})$ by naively assuming the clients are truthful, e.g.  $Q_{it}(u_{it}|r_{it}) = \textbf{1}_{r_{it}}$. The Greedy Mechanism,  which uses the greedy algorithm as its allocation rule, is truthful when we assume that each client plays requests from the same fixed distribution $Q_{i}$ at each round or when the set of request distributions are determined non-adaptively. This is easy to see because the distributions played by each client are fixed a priori at every round, and by Theorem~\ref{truthfulness} each auction at each round is truthful.

It is not the case however, that the Greedy Mechanism is even approximately truthful when the players have the ability to arbitrarily adapt their strategies over a series of rounds. Example~\ref{example} in the Appendix demonstrates that truthfulness is violated because over multiple rounds adaptivity allows clients to potentially coordinate their behavior. It turns out that in addition to providing privacy, JDP is precisely the property that makes the general form of the allocation mechanism truthful when clients are only concerned with maximizing their own utilities, by limiting the ability of client's requests to influence the requests of another client (and thereby coordinate) across rounds. 

\begin{theorem}[Approximate Truthfulness]
\label{thm:apxtrue}
Let $A$ be the allocation rule $\pauction(\alpha, U, V, \epsilon, \rho)$ such that $A$ is $(\epsilon', \beta/T)$-JDP with $\epsilon' = \tilde{O}(\epsilon/\sqrt{T})$ 
 and outputs $S$ such that $\E[V(S)] \geq (1-\rho)OPT_{V}-\rho$. Take $\beta, \rho$ such that $\sqrt{\beta + (1-\beta)\rho} \leq \beta^2/T$.
 Then for a $(1-\beta)$ fraction of the $n$ clients $i$: \\
 Let $L_{i*}^{t}$ denote the truthful strategies, and let $L_{i}^{t}$ be any other set of strategies. Algorithm~\ref{mech:priv} with allocation rule
$\pauction$ satisfies:
$$v_i(L_{i}^{1}, \ldots, L_{i}^{n}) \leq e^{2\epsilon}v_i(L_{i*}^{1}, \ldots, L_{i*}^{n}) + 2\beta UT + e^{\epsilon}\frac{\beta^2}{1-\beta^2/T}$$
\end{theorem}
We note that $\epsilon'$ is set to ensure $(\epsilon, \beta)$-privacy after $T$ rounds of composition using the advanced composition theorem of \cite{fast}. Now conditional on all clients $i$ requesting truthfully at all rounds $t$, the Private Greedy Mechanism consists of $T$ runs of $\pauction$ with input distributions $Q_i = \textbf{1}_{r_{it}}$, the true posterior on $u_{it} = r_{it}$. Hence with the same settings as Theorem~\ref{thm:apxtrue}:

\begin{theorem}
\label{still:opt}
$$v_{\A}(L_{i*}^t) \geq (1-\rho)OPT_V -\rho T,$$
where $OPT_{V}$ denotes the lender's optimal utility.
\end{theorem}
\sn{perhaps we should be more explicit about how these results only hold when $n$ is sufficiently large?}

\begin{remark}[Learning]
The only drawback to our Private Greedy Mechanism is that the learning component in step $8$ is trivial; it only learns correctly when clients are truthful. While this is enough to incentivize truthfulness, it would fail if clients were instead to play from a fixed (dishonest) distribution $Q_{it} = Q_i$. A more practical algorithm would try to learn the distributions while preserving truthfulness, and by learning that the distributions obtain approximately optimal utility for the lender against arbitrary input distributions $Q_i$. The only property the learning step $8$ has that was used to establish the above results, is that if clients do report truthfully at every round, the posterior on $u_{it}$ should be $\textbf{1}_{u_{it}}$ for all $i,t$. Consider the naive learning algorithm that for each arm $i$ for a period of $T$ rounds assigns all the shares (up to $U$ the maximum demand) to arm $i$ and observes an uncensored observation $u_i$. After drawing enough samples to estimate each conditional $Q_i(u|r_i)$ for each client $i$ and each possible request, the algorithm computes the optimal allocation with respect to these estimates using \emph{Greedy}. This algorithm clearly has the desired property, and in the Appendix we sketch a proof via standard concentration arguments that in polynomially many rounds we can estimate the conditional distributions $Q_i(u|r_i)$ well enough to allocate approximately optimally using \emph{Greedy}.
\end{remark}

\bibliographystyle{alpha}
\bibliography{neurips_2019}
\appendix
\section{Preliminaries and Privacy Basics}
\begin{algorithm}
\caption{Securities Lending Mechanism}
\begin{algorithmic}[1]
\label{sec:lend}
\Procedure{$\A$}{$\{U_i\}_{i=1}^{n}$, $V$ shares, time horizon $T$}
\For{$t = 1\ldots T$}
       \For{i = 1 \ldots n}
		\State Client $i$ draws $u_{it} \sim U_i$
		\State Client $i$ picks request distribution $Q_{it} = L_t^i(\mathcal{H}_{t}^{i}, u_{it})$ 
		\State Client $i$ draws $r_{it} \sim Q_{it}$, and submits $r_{it}$
	\EndFor
	\State $\A$ computes an allocation $S_t = \A(r_{1t}, \ldots r_{nt}, H_t)$
	\State $\A$ observes the executed shares $v_i(S_t)$ for each client
	\State $\A$ updates the history: $H_{t+1} = H_t \cup (r_{it}, s_{it}, v_i(S_t))_{i=1}^n$
\EndFor
\EndProcedure
\end{algorithmic}
\end{algorithm}

\begin{lemma}[Post Processing \cite{DMNS06}]
\label{post}
Let $A:\mathcal{X}^n\rightarrow \mathcal{O}$ be any $(\eps,\delta)$-differentially private algorithm, and let $f:\mathcal{O}\rightarrow \mathcal{O'}$ be any (possibly randomized) algorithm. Then the algorithm $f \circ A: \mathcal{X}^n \rightarrow \mathcal{O}'$ is also $(\eps,\delta)$-differentially private.
\end{lemma}
Post-processing implies that, for example, every \emph{decision} process based on the output of a differentially private algorithm is also differentially private.

\begin{theorem}[Composition \cite{DMNS06}]\label{composition}
Let $A_1:\mathcal{X}^n\rightarrow \mathcal{O}$, $A_2:\mathcal{X}^n\rightarrow \mathcal{O}'$ be algorithms that are $(\eps_1,\delta_1)$- and $(\eps_2,\delta_2)$-differentially private, respectively. Then the algorithm $A:\mathcal{X}^n\rightarrow \mathcal{O}\times \mathcal{O'}$ defined as $A(r) = (A_1(r), A_2(r))$ is $(\eps_1+\eps_2,\delta_1+\delta_2)$-differentially private. This holds even if $A_2$ may be chosen as a \emph{function} of the output of $A_1$.
\end{theorem}

\begin{lemma}[Billboard Lemma \cite{privmatchings}]
\label{lem:billboard}
Let $M:\mathcal{X}^n\rightarrow \mathcal{R}$ be an $(\epsilon,\delta)$-differentially private algorithm. Let $A:\mathcal{X}^n\rightarrow \mathcal{O}^n$ be any algorithm that can be decomposed as follows:
\begin{enumerate}
\item On input $r$, compute $p = M(r)$.
\item Output $f_i(r_{it},p)$ to agent $i$, where $f_i:\mathcal{X}\times \mathcal{R}\rightarrow \mathcal{O}$ is an arbitrary function.
\end{enumerate}
Then $A$ is $(\epsilon,\delta)$-jointly differentially private.
\end{lemma}

\section{Proofs from Section~\ref{sec:offline}}
\textbf{Proof of Theorem~\ref{thm:greedy1}.}
\begin{proof}
	We first observe that the tail probabilities are monotonically decreasing in s for each client $i$: $T_i(s|r_i) \geq T_i(s'|r_i)$ for all $s \leq s'$.  Therefore, by greedily allocating shares to clients in decreasing order of $T_i(s|r_i)$, \emph{Greedy} returns
	
	$$ S \in \arg\max  \sum_{i=1}^{n} \sum_{s=1}^{s_i} T_i(s|r_i)\ s.t.\ \sum_{i=1}^{n} s_i = V$$
	
	It remains to show that the expression above is equivalent to the expected number of units used. For an arbitrary client $i$:

\begin{eqnarray*}
&& 	\mathbb{E}_{u_i \sim Q_i(u|r_i)}[\min(s_i, u_i)] \\
 &=& \sum_{u=0}^{U} Q_i(u|r_i) \min(s_i, u) \\
 &=& \sum_{u=0}^{U} uQ_i(u|r_i)\textbf{1}\{u < s_i\} + \sum_{u=0}^{U} s_i Q_i(u|r_i)\textbf{1}\{u \geq s_i\} \\
	&=& \sum_{u=0}^{s_i-1} uQ_i(u|r_i) + \sum_{u=s_i}^{U} s_i Q_i(u|r_i) \\
&=& \sum_{u=0}^{s_i-1} uQ_i(u|r_i) + s_iT_i(s_i|r_i) \\
 &=& \sum_{u=0}^{s_i-2} uQ_i(u|r_i) + (s_i-1)Q_i(s_i-1|r_i) + s_iT_i(s_i|r_i) \\
 &=&  \sum_{u=0}^{s_i-2} uQ_i(u|r_i) + (s_i-1)Q_i(s_i-1|r_i) + T_i(s_i|r_i) + (s_i-1)(T_i(s_i-1 | r_i) - Q_i(s_i-1|r_i)) \\
 &=&  \sum_{u=0}^{s_i-2} uQ_i(u|r_i)  + (s_i-1)T_i(s_i-1 | r_i) + T_i(s_i|r_i)
 \end{eqnarray*}
Here, in the penultimate line, we use the fact that $T_i(s_i | r_i) = T_i(s_i-1 | r_i) - Q_i(s_i-1|r_i)$. Continuing this manipulation inductively, we obtain that:
	$$ \mathbb{E}_{u_i \sim Q_i(u|r_i)}[\min(s_i, u_i)] = \sum_{s=1}^{s_i} T_i(s| r_i).$$
	

	Thus, $\sum_{i=1}^{n} \sum_{s=1}^{s_i} T_i(s|r_i) = \sum_i \mathbb{E}_{Q_i(u|r_i)}[\min(s_i, u_i)]=V(S)$, the expected payoff to the lender.

\end{proof}

\textbf{Proof of Theorem~\ref{truthfulness}.}
\begin{proof}
Recall the tail probabilities, $T_i(s|r_i):=\Prob{u_i \sim Q(u|r_i)}{u \geq s} = \sum_{s' \geq s}Q_i(s'|r_i)$, and consider the set $\{T_i(s_i|r_i)\}_{i=1}^{n}$ where $s_i$ ranges from 1 to $U$. Let $T_{(k)}$ be the k-th order statistic of the set $\{T_i(\cdot|r_i)\}$, where $|\{T_i(\cdot|r_i)\}|=M$, and let $\Omega:= \cup_{k=0}^{V-1} \{T_{(M-k)}\}$, the set of the V largest tail probabilities among all clients for all $ s \geq 1$. The claim follows from observing that the allocation strategy employed by the lender is equivalent to giving client $i$ shares equal to the number of times one of her tail probabilities appears in $\Omega$, described below.


When a truthful client $i$ requests $r_i$ shares, the tail probabilities of using $s$ shares will be

$$T^*_i(s|r_i)=\sum_{s'\geq s} Q_i(s'|r_i)=\sum_{s' \geq s}\textbf{1}\{s\leq u_i\}=\textbf{1}\{s\leq u_i\}$$

Let $T_i (s|r_i)$ denote the tail probabilities given a draw $r_i \sim Q_i(r_i|u_i)$ from an arbitrary strategy $Q_i$.
Now fix any draw $r_i \sim Q_i(r_i|u_i)$ from $Q_i$. Conditioned on $r_i$, the number of shares allocated to client $i$ can be written as:
$$
s_i|r_i = \sum_{s=1}^{U} \textbf{1} \{T_i(s|r_i) \in \Omega\},
$$
e.g. the number of tail probabilities of client $i$ that are among the $V$ largest. Then the utility
$v_i =\min(u_i, s_i)=\sum_{s= 1}^U \textbf{1} \{T_i(s|r_i) \in \Omega\} \textbf{1}\{1 \leq s \leq u_i\} = \sum_{s= 1}^{u_i} \textbf{1} \{T_i(s|r_i) \in A\}$. This of course holds for the truthful strategy as well, and so it suffices to show that:
$$
\sum_{s= 1}^{u_i} \textbf{1} \{T^*_i(s|u_i) \in \Omega\} \geq \sum_{s= 1}^{u_i} \textbf{1} \{T_i(s|r_i) \in \Omega\}
$$
This holds immediately, since  $\forall s \leq u_i, T_i^*(s|u_i) = 1 \geq T_i(s|r_i),$ which implies $\textbf{1} \{T^*_i(s|u_i) \in \Omega\} \geq \textbf{1} \{T_i(s|r_i) \in \Omega\}$. So we've shown that for any $q_i$ and $r_i \sim Q_i(r_i|u_i), v_i(A(r_i, r_{-i} ; Q_i, Q_{-i}))] \leq   v_i(A(u_i, r_{-i} ; Q^T_i, Q_{-i}))$. Since this holds for any fixed $r_i$, it holds when we take the expectation over $r_i \sim Q_i(r_i|u_i)$, proving the claim.

\end{proof}

\section{Proofs and Definitions from Section~\ref{sec:game}}

\begin{algorithm}[h]
\label{alg:auc}
  \caption{\Auction Rule}\label{auction1}
  \hspace*{\algorithmicindent} \textbf{Input:} $\alpha > 0, n, \{v_i\}_{i \in [n]}, U, V$ \Comment{valuations $v_i: [U] \to [0,1]$ satisfy DMR property}\\
 \hspace*{\algorithmicindent} \textbf{Output:} feasible allocation $S$.
  \begin{algorithmic}[1]
    \Procedure{\Auction}{$\alpha, U, V$}
      \State Initialize array $S$ of length $n$, $S[i] \gets 0, \forall i$ \Comment{$S[i]$ counts the goods currently allocated to player $i$}
      \State Initialize $\cB \gets n, T_B \gets 0$ \Comment{counts number of bids in current round, total bids respectively}
      \State Set the price $p = 0$, $m = 1$ \Comment{$m$ is the index of the good currently being allocated}
      \While{$\cB \not=0$}\Comment{terminate if there are $0$ bids in the round}
        \State $\cB \gets 0$
  	  \For{$i = 1 \ldots n$}
	  \State Let $\Delta_i = v_i(S[i]+1)-v_i(S[i])$ \Comment{marginal utility of additional good}
        		\If{$\Delta_i \geq p$}
			\State $\cB \gets \cB + 1, S[i] \gets S[i] + 1, m \gets m + 1$ \Comment{when $m = V + 1$, set $m = 1$}
			\State $S[i_m] \gets S[i_m] -1 $, where $i_m$ is the player to which good $m$ is currently allocated
			\If{$T_B \pmod{V} = 0$} \Comment{increment the price every $V$ bids}
				\State $p \gets p + \alpha$
			\EndIf
       		 \EndIf
      \EndFor
      \EndWhile\label{euclidendwhile}
      \State \textbf{return} $S$
    \EndProcedure
  \end{algorithmic}
\end{algorithm}

\textbf{Proof of Theorem~\ref{thm:auction1}.}

\begin{proof}
We first show termination, then accuracy. \\
\textbf{Termination:} For every round that is not the final round, the total number of bids $T_B$ increases by at least $1$; otherwise, the algorithm would have terminated. Hence after every $V$ such rounds, the price increases by at least one increment of $\alpha$. When $p \geq 1$, no further bids are placed because $\Delta_i \leq 1$ ($\Delta_i  = v_i(S[i] + 1)- v_i(S[i])$ as in line $8$ of $\Auction$)\ar{Haven't seen the $\Delta_i$ notation before. Make sure its defined if we use it.}. This occurs after at most $V/\alpha$ many rounds. After the next round the algorithm necessarily terminates. \\
\textbf{Accuracy}:
We will show that at termination, the price $p^{*}$ in conjunction with the valuation functions $\{v_i\}$ form an \textit{approximate Walrasian equilibrium}. For our purposes this will mean that:
\begin{enumerate}
\item The allocation $S$ is feasible and all goods are allocated; i.e. $\sum_i S[i] = V$
\item Each player $i$ receives her approximately most preferred allocation at the current price level:
\end{enumerate}
$$
v_i(S[i])-p^{*}S[i] \geq (\max_{l \in [U]}v_i(l)-lp^*)-S[i]\alpha
$$

The first condition follows by construction and from the fact that $v_i \geq 0$; so, when the price is $0$, players will continue bidding until all goods are allocated, at which point the price is incremented for the first time. Since the number of goods allocated is non-decreasing as the auction is run, this is enough to conclude that $\sum_i S[i] = V$. The second condition forms the bulk of the proof and is where the DMR property will be used in a critical way. But deferring that proof, let us see how this implies the theorem statement. Let $S'$ be any feasible allocation. Then

$$
v_i(S[i])-p^{*}S[i] \geq v_i(S'[l])-S'[l]p^*)-S[i]\alpha \; \forall i \rightarrow
$$

$$
\sum_i v_i(S[i])-p^{*}V \geq \sum_i(v_i(S'[l])-S'[l]p^*)- V\alpha = \sum_i(v_i(S'[l])) -\sum_i(S'[l]p^*)- V\alpha
$$
Since $S'$ is feasible, $\sum_i(S'[l]p^*) \leq Vp^*$, and hence $\sum_i v_i(S[i])-p^{*}V \geq \sum_i(v_i(S'[l])) -Vp^*- V\alpha$,
which implies $\frac{1}{n}\sum_i v_i(S[i]) \geq \frac{1}{n}\sum_i(v_i(S'[l])) - \frac{V\alpha}{n}$. This holds for any allocation $S'$, and so it holds for the optimal allocation, which proves the claim. So it suffices to show $2.$\\

Fix a player $i$. Suppose that $S[i] = 0$. Then we need to show that $0 \geq \max_{l \in [U]}(v_i(l)-lp^{*})$. We can rewrite this as:
$\max_{l \in [U]}\sum_{j = 1}^{l}(v_i(j)-v_i(j-1)-p^{*}),$ since canonically $v_i(0) = 0$. Moreover, since $i$ declined to bid for $1$ share at price $p^{*}$ (since the \Auction has terminated) we know $v_i(1) = v_i(1)-v_i(0)-p^{*} < 0$. But by the DMR property, this implies that $v_i(j)-v_i(j-1)-p^* <0$, for $j \geq 1$. Hence $0 \geq \max_{l \in [U]}(v_i(l)-lp^{*})$ since for any value of $l$, each term in the sum on the RHS is negative.

Now suppose $S[i] > 0$. Let $p_i$ be the price at which player $i$ last bid. Let $\tilde{s}$ be the number of copies of the good that player $i$ had at that time. Because $i$ has not gained any shares since her last bid, $\tilde{s} \geq S[i]-1$.
First we observe that $p_i \geq p^*-\alpha$, since the price only increments every $V$ bids, and in that time, all of the goods are re-allocated to a new player. Since $S[i] > 0$, player $i$ must have bid and received a good at some point during the last $V$ bids, which means the price can only have incremented at most once. Suppose now that $l > S[i]$. Since $i$ did not bid at price $p^{*}$, we know that $v_i(S[i]+1)-v_i(S[i]) < p^*.$ By DMR, this means that for all $j \geq S[i]+1, v_i(j)-v_i(j-1) < p^*$. Thus $v_i(l)-v_i(S[i]) = \sum_{j = S[i]+1}^{l}v_i(j)-v_i(j-1) \leq (l-S[i])p^*$, which rearranging shows that
$v_i(l)-p^*l \leq v_i(S[i])-p^*S[i]$, which satisfies condition $2$. Now consider the case where $l < S[i]$. We know that at price $p_i \geq p^*-\alpha$, with $\tilde{s} \geq S[i]-1$ copies of the good, player $i$ bid for an additional share.

\begin{lemma}{Monotonicity of Valuations.}
\label{monotonic}
Suppose that player $i$ with $s$ copies of the good bids for an additional good at price $p$. Let $j \leq t \leq s+1$.
Then $v_i(t)-tp \geq (v_i(j)-jp)$.
\end{lemma}
\begin{proof}
Player $i$ bids at price $p$, and so we know that $v_i(s+1)-v_i(s) \geq p$. By the DMR property, this means that
$v_i(k+1)-v_i(k) \geq p$, for $k \leq s$. Since $v_i(t)-v_i(j) = \sum_{k = j+1}^{t}v_i(k)-v_i(k-1)$, this shows that $v_i(t)-v_i(j) \geq p(t-j)$. Rearranging proves the lemma.
\end{proof}

By Lemma~\ref{monotonic}, and the fact that $S[i] \leq \tilde{s} +1$,  this means that $v_i(S[i])-S[i]p_i \geq v_i(l)-lp_i$. Rearranging, and using that $p_i \geq p^{*}-\alpha$, this shows that
$v_i(S[i])-p^*S[i] - (v_i(l)-p^*l) \geq -(S[i]-l)\alpha \geq -S[i]\alpha$, since $l \leq S[i]$. This proves the claim.

\end{proof}

\section{Proofs and Definitions from Section~\ref{sec:auction}}
\begin{algorithm}[h]
  \caption{\pauction}\label{auction2}
  \hspace*{\algorithmicindent} \textbf{Input:} $\alpha > 0, n, \{v_i\}_{i \in [n]}, U, V, \epsilon, \rho$ \Comment{valuations $v_i: [U] \to [0,1]$ satisfy DMR property}
 \hspace*{\algorithmicindent} \textbf{Output:} Allocation $S = \{s_i\}$
  \begin{algorithmic}[1]
    \Procedure{\pauction}{$\alpha, U, V, \epsilon, \rho$}
      \State Set $T = \frac{2V}{\alpha \rho n}$
      \State For each $i$, let $S_i$ be the set of shares allocated to $i$. Initialize $S_i = \emptyset$
      \State For $s \in S_i$, let $s[count]$ counts the number of shares allocated since $i$ received $s$ \ar{Can we define this just in terms of the bid count? For JDP, it is important we don't access ``the number of shares allocated'' except via the bid count.}
      \State Initialize private counter $C_{\epsilon'}(nT)$
      \State Let $\epsilon' = \frac{\epsilon(\alpha \rho n)}{V}$
      \State Set $round_b = n$, $E = \frac{2\sqrt{2}\log(1/\beta)\log(nT)^{5/2}}{\epsilon'}$ \Comment{noisy round bid count,  error of $C_\epsilon(nT)$}
      \State Set $V' = V-2E$ \Comment{We leave some slack in the supply $V'$ so that our allocation is always feasible}
      \While{$ round_b > \rho n - 2E$ and $t \leq T$}\Comment{terminate if there are $0$ bids in the round}
  	  \For{$i = 1 \ldots n$}
	  \State $p = \alpha \lfloor \frac{C_{\epsilon'}[t]}{V'}\rfloor$ \Comment{estimate the current price}
	  \State $\Delta = C_{\epsilon'}[t]-C_{\epsilon'}[t-n]$ \Comment{estimate the number of shares allocated since last bid}
	  \For{shares $s \in S_i$}
		\State update $s[count] \gets s[count] + \Delta$
		\If{$s[count] \geq V'$}
			\State $S \gets S \backslash \{ s \}$
		\EndIf
	\EndFor
	  \State Set $b_i = 0$ \Comment{$b_i$ is the indicator if $i$ bids}
        		\If{$v_i(|S_i|+1)-v_i(|S_i|) \geq p$}
			\State $S \gets S \cup \{s\}$
			\State Initialize $s[count] = 0$			
			\State Feed $b_i$ to $C_{\epsilon'}(nT)$
       		\EndIf
		\State $t \gets t + 1$
      \EndFor
      \State $round_b \gets C_{\epsilon'}[t]-C_{\epsilon'}[t-n]$ \Comment{noisy count of bids in the round}
      \EndWhile\label{euclidendwhile}
      \State \textbf{return} $S = \{|S_i|\}$ \ar{Funny type mismatch. As we have defined it, an allocation specifies a number of shares for each person. But here, each $S_i$ is a set.}
    \EndProcedure
  \end{algorithmic}
\end{algorithm}

Given a stream of bits (which in our case will represent ``bids'') $b = (b_1, b_2, \ldots b_T) \in \{0,1\}^{T}$, a streaming counter $C(b)$ releases an approximation $C(b)[t]$ to $s_b(t) = \sum_{i = 1}^{t}b_i$ at every time step $t$.
\begin{definition}{\cite{chan}}
A streaming counter $C$ is $(\alpha, \beta)$ useful if with probability at least $1-\beta$, for each $t \in [T]$,
$$
|C(b)[t]-s_b(t)| \leq \alpha
$$
\end{definition}

We will denote by $C_\epsilon(T)$ the Binary mechanism of \cite{chan}, instantiated with parameter $\epsilon$ and time horizon $T$. \ar{Do we want a variant that maintains a monotonically increasing count?}
\begin{theorem}{\cite{chan}}
For $\beta > 0$, and any sequence  $b$, $C_\epsilon(T)$ is $\epsilon$-differentially private with respect to a change in a single entry of the stream $b$
and $(\alpha, \beta)$-useful for
$$
\alpha = \frac{2\sqrt{2}\log(1/\beta)\log(T)^{5/2}}{\epsilon}
$$
\end{theorem}

\textbf{Proof of Theorem~\ref{privacc}.}
\begin{proof}
\textbf{Accuracy analysis:} Algorithm $\pauction$ potentially loses welfare (as compared to the allocation computed by its non private variant) in three ways:
\begin{enumerate}
\item To ensure a feasible allocation, the auction only tries to allocate supply $V-2E$ and might potentially allocate only $V - 4E$ shares.
\item The auction stops after the first round in which fewer than $\rho n$ players bid, rather than continuing to termination.
\item Prices and allocations are computed with respect to a noisy estimate of bid counts, rather than with respect to exact counts.
\end{enumerate}
We handle the first source of error with the following three lemmata.

\begin{lemma}
\label{multappx}
Assume that all valuation functions $v_i$ have the DMR property. Then
$OPT_{V-E} \geq (1-E/V)OPT_V$.
\end{lemma}
\begin{proof}
Consider the optimal allocation of $V$ shares, $S_V$. The welfare of $S_V$ can be written as the sum over all $V$ shares of the marginal value of allocating that share. I.e. for each player $i$ assigned $s_i$ shares:
$v_i(s_i) = \sum_{j = 1}^{s_i}(v_i(j)-v_i(j-1))$, and the welfare of $S_V$, $OPT_V = w(S_v) = \sum_i v_i(s_i) = \sum_i \sum_{j = 1}^{s_i}(v_i(j)-v_i(j-1))$. Now consider the $E$ shares with the lowest marginal values, and write $S_E$ to denote the allocation of these $E$ shares. By the DMR property, there exists a feasible allocation $S_{V-E}$ of $V-E$ shares respectively such that $S_V = S_{V-E} + S_E$.  Then, $w(S_V) = w(S_{V-E}) + w(S_E) \leq OPT_{V-E} + w(S_E) \rightarrow OPT_{V-E} \geq w(S_{V})-w(S_E) = w(S_V)(1-w(S_E)/w(S_V))$. Since by definition $w(S_E)$ is the sum of the $E$ lowest marginal values of the shares $V$ in the allocation $S_V$, $w(S_E)/w(S_V) \leq E/V$, which proves the claim.
\end{proof}

\begin{lemma}[Feasibility.]
\label{lem:feasibility}
With probability at least $1-\beta$, if $S$ is the allocation returned by $\pauction$, $|S| = \sum s_i \leq V$.
\end{lemma}
\begin{proof}
With probability $1-\beta$, we know that the error of the private bid counting sequence $C_{\epsilon'}(nT)$ is less than $E$ for every time step:
\begin{equation}
\label{err}
\sup_{t = 1 \ldots nT}|C_{\epsilon'}[t]-s_b[t]| \leq E,
\end{equation}
where again $s_b[t]$ denots the true (non-noisy) bid count at time $t$.

We claim then that if (\ref{err}) holds, then for every allocated share $s$  allocated to any client $i$ at any time $t$, the error in $s[count]$, which counts the number of bids since $s$ was allocated, is at most $2E$. We first write $s[count] = C_{\epsilon}[t] - C_{\epsilon}[t-ri]$, where $r$ is the number of rounds since $s$ was allocated. The error of $s[count]$ is:
 $$s[count] - s_b[t]-s_b[t-ri] \leq |C_{\epsilon}[t]-s_b[t]| + |C_{\epsilon}[t-ri]-s_b[t-ri]| \leq 2E$$

Now consider the time at termination $T*$, and let $t^* < T*$ denote the time at which the $V^{th}$ last bid was made. For any share $s$ allocated at any time $t < t^*$, at $T*$ there have been greater than $ V $ shares allocated since $s$ was allocated. Hence, $s[count] \geq V- 2E  = V'$. Since a share is unallocated whenever its estimate of the bids placed since it was allocated, $s[count]$,  exceeds $V'$, we know that shares allocated prior to the last $V$ bids have all been unallocated by $T*$. Thus, the number of shares allocated at termination is upper bounded by $V$.
\end{proof}

\ar{This is the kind of thing that can be broken out as a lemma}
\begin{lemma}[Approximate Clearing]
\label{appx}
With probability at least $1-\beta$, if $S$ is the allocation returned by $\pauction$, $|S| = \sum s_i \geq V-4E$.
\end{lemma}
\begin{proof}
Again let $T^*$ be the time at termination, and note that (\ref{err}) holds with probability at least $1-\beta$. We first show that the total number of bids, $s_b[T^*] \geq V-3E$. This is clear, because for the first $t = 1 \ldots V-3E$ rounds, the counter $C_\epsilon[t] \leq V-3E + E = V'$, and hence $p = 0$. At $p = 0$ every player $i$ bids. Hence, since $s_b[T] \geq V-3E$, we can consider the last $V-4E$ bids over the course of the auction. By definition, the true number of bids that have been submitted since any of these shares has been allocated is less than $V-4E$. Moreover, for each of these allocated shares $s$ the error of $s[count]$ is less than $2E$. And so, at $T^*, s[count] \leq V-3E + 2E = V'$; therefore, none of these shares have been unallocated. Thus, there are at least $V-4E$ shares allocated at termination.
\end{proof}

Let $p^* = \alpha\lfloor \frac{C_\epsilon'[T^*]}{V'} \rfloor$ denote the price at termination. We say that a client $i$ is \emph{unsatisfied} if at $T^*$, client $i$ would still bid; e.g. if $v_i(S_i + 1)- v_i(S_i) \geq p^*$.  Then, we claim that under (\ref{err}) at $T^*$ there are at most $\rho n$ unsatisfied clients:
\begin{itemize}
\item If the algorithm terminates early, then $round_b < \rho n - 2E$, which means that the number of bids in the last round, $s_b[T^*]-s_b[T^*-n]$ is at most $\rho n$.
\item If the algorithm does not terminate early, then after $T$ rounds, there have been at least $\rho n - 4E$ bids per round, for a total of $T(\rho n-4E)$ bids. Thus, the price is at least $\frac{\alpha T(\rho n-4E)}{V}$. By assumption, $E \leq \rho n/8$, substituting $T = \frac{2V}{\alpha \rho n}$ shows that the price is $\geq 1$, and there are no bidders.
\end{itemize}

We now show that under the allocation $S$, each satisfied player receives her approximately most preferred bundle at the current price level. If $I^*$ denotes the at least $(1-\rho)$ fraction of satisfied bidders at termination, then for $i \in I^*$ we show:

\begin{equation}\label{appxopt}
v_i(S[i]) - p^*S[i] \geq \max_{l \in [U]}v_i(l)-p^*S[i]- S[i](2\alpha)
\end{equation}

This is very similar to the argument used to show optimality of the non-private ascending price auction. If $S[i] = 0$ the statement holds trivially, because the fact that $i$ has chosen not to bid at price $p^*$ means all the bundles on the right hand side have negative value, by the DMR property.
So, suppose that $S[i] > 0$. Let $p_i$ be the price at which player $i$ last bid. Again, whenever $V$ bids go by, since the error of the counter is less than $E$, every share must have been reallocated. Since $S[i] > 0$, this means $i$ last bid within the last $V$ bids, which in turn means that the counter could have incremented by at most $V + E$ counts since $i$ last bid. In turn, this means that the price could have gone up by at most $\alpha + \lceil \frac{E}{V} \rceil \alpha \leq 2\alpha$. So $p_i \geq p^*-2\alpha$. Following an identical chain of reasoning as in the proof of Theorem~\ref{thm:auction1}, we recover Equation~\ref{appxopt}. Now let $S'$ be any other feasible allocation of $V'$ shares. Then

$$
v_i(S[i])-p^*S[i] \geq v_i(S'[l])-S'[l]p^*)-S[i]2\alpha \; \forall i \in I^* \rightarrow
$$

$$
(1-\rho)n + \sum_{i \in S*} v_i(S[i])-\sum_{i}S[i]p^* \geq \sum_{i=1}^n(v_i(S'[l])-S'[l](p^* + 2\alpha))
$$
where the second line follows since $v_i \leq 1$.
Since $S'$ is feasible, $\sum_i(S'[l]p^*) \leq V'p^*$. Moreover, since at least $V-4E$ shares are allocated at termination,
$$
\rho n + \sum_{i \in S*} v_i(S[i])-(V-4E)p^* \geq \rho n + \sum_{i \in S*} v_i(S[i])-\sum_{i}S[i]p^* \geq \sum_{i=1}^n(v_i(S'[l]))-V'(p^* + 2\alpha))
$$
Rearranging, and using $p^*\leq 1$, this shows that:
$$
v(S) \geq v(S') - \rho - \frac{2(1-2\alpha)E + 2\alpha V}{n}
$$

Since this holds for any allocation $S'$, it certainly holds for the optimal allocation.
Thus, $v(S) \geq OPT_{V-2E} - (1-\rho) - \frac{2(1-2\alpha)E + 2\alpha V}{n}$, which by Lemma~\ref{multappx}, means $$v(S) \geq (1-2E/V)OPT_V- (1-\rho) - \frac{2(1-2\alpha)E + 2\alpha V}{n}$$
The result then follows from the numbered conditions in the theorem statement.  \\

\textbf{Privacy Analysis:}
With probability $1-\beta$ over the randomness in the private bid counter, there are at least $\rho n - 4E$ bids at
every round of the algorithm. Moreover, after the noisy price is $1$, there are no bids. Hence, after $\frac{V}{\alpha} + E$ total bids, the noisy counter reaches $\frac{V}{\alpha}$, and bidding ends. Thus, with probability $1-\beta$ there are at most $\frac{\frac{V}{\alpha} + E}{\rho n - 4E}$ rounds, which means that each player bids at most that many times. Thus, the sensitivity of the counter is $\frac{\frac{V}{\alpha} + E}{\rho n - 4E}$. Since we take $\alpha \leq \frac{V}{\rho n}$ (Condition $1$), $\frac{\frac{V}{\alpha} + E}{\rho n - 4E} \leq \frac{2V}{\alpha \rho n}$. Hence, with probability $1-\beta$, setting $\epsilon' = \frac{\epsilon}{\frac{2V}{\alpha \rho n}}$ guarantees $\epsilon$ differential privacy \cite{chan}. Since this only happens with probability $1-\beta$, this guarantees that $C_{\epsilon'}(nT)$ satisfies $(\epsilon,\beta)$-differential privacy. Since each individual allocation $S_i$ is computed purely as a function of $C_{\epsilon'}(nT)$, by the billboard lemma, \pauction achieves $(\beta, \epsilon)$-JDP. \ar{Would be nicer just to have people stop bidding if they are ever asked to make too many bids. This bounds the sensitivity of the running sum with probability 1, so we don't have to state an $(\epsilon,\delta)$ privacy bound, and pulls the $\beta$ failure probability into the utility analysis.} We remark that if we are concerned with achieving the stronger $(\epsilon, 0)$-JDP, then instead of concluding that with high probability no client bids more than $\frac{V/\alpha + E}{\rho n - 4E}$ many times, we could instead modify $\pauction$ in such a way that after a client bids $\frac{V/\alpha + E}{\rho n - 4E}$ many times, they automatically stop bidding. In this case, the sensitivity of the bid counter with respect to any client will always be bounded by $\frac{V/\alpha + E}{\rho n - 4E}$, and the counter of \cite{chan} achieves $(\epsilon, 0)$ differential privacy, which implies that $\pauction$ achieves $(\epsilon, 0)$-JDP by Lemma~\ref{lem:billboard}. We note that our accuracy statement will be unchanged, since in the analysis we are already conditioning on the fact that the error is bounded by $E$, whereby no client exceeds $\frac{V/\alpha + E}{\rho n - 4E}$ bids.
\end{proof}

\textbf{Proof of Theorem~\ref{thm:privtrue}.}
\begin{proof}
We first condition on the private trajectory of bids $C_{\epsilon'}(nT)$, which induces a private trajectory of prices $\vec{p}$. Then, given the trajectory of prices, the client utilities $v^i_A(\cdot)$ are deterministic quantities. We claim that for all $\vec{p}$,
$v^i_A(Q^T|\vec{p}) \geq v^i_A(Q_i|\vec{p})$. The claim will follow from the fact that, given any fixed trajectory $\vec{p}$, we imagine running two versions of $\pauction$, one where client $i$ reports $Q_i(r|u_i)$ and the other where client $i$ reports $Q^T$. Then, restricting our attention to the shares $s$ that $i$ has at termination, we will show that in any round in the $Q_i$ auction where $s$ was acquired, $s$ was also acquired in the $Q^T$ auction. Hence, the $Q^T$ auction results in at least as many shares at termination as the $Q_i$ auction, which since the client utility is monotonically increasing in the number of shares $s_i$, shows the claim.

In the next paragraph we stop explicitly writing the conditioning on $\vec{p}$ and assume a fixed sequence of prices.
Then, we analyze $2$ cases:
\begin{enumerate}
\item The final price $p^*$ at termination is $\leq 1$.
\item The final price $p^*$ at termination $1 + \alpha$

\end{enumerate}
\textbf{Case 1:}  Assume that $i$ is one of the $(1-\rho)n$ fraction of bidders who is satisfied \ar{Defined?} at $p^*$. Then, it is immediately clear that if $\vec{p}$ falls into case $1$, then $v^i_A(Q^T) = u_i$. \ar{confusing notation since $p$ is price} This follows from the fact that the marginal value under $Q^T$ of every additional share is $1$ up until share $u_i + 1$, and so if the price is $\leq 1$, client $i$ will bid until she has at least $s = u_i$ shares. Since this maximizes her realized payoff function $\min(u_i, s)$, in case $1$ we have $v^i_A(Q^T|\vec{p}) \geq v^i_A(Q_i|\vec{p})$. \\

\textbf{Case 2:}
We first note that since we have drawn $u_i, r_i$, this implies that $Q_i(r_i|u_i) > 0, Q_i(u_i) > 0, Q_i(r_i) > 0$, and hence $Q_i(u_i|r_i) = \frac{Q_i(r_i|u_i)Q_i(u_i)}{Q_i(r_i)} > 0$. Thus, under $Q_i$ the marginal value of the $j^{th}$ share, $\Prob{u_i \sim Q_i(u|r_i)}{u_i \geq j}$, is less than $1$ for $j > u$. Then, as we have used repeatedly, any share $s$ that client $i$ holds at termination must have been acquired at a price $p_i \geq p^* - \alpha = 1$, since the price increments by $\alpha$ after only $V'$ ticks of the noisy counter. But, since the marginal value of every share beyond the $u^{th}$ is strictly less than $1$, any share $s$ held at termination in the $Q_i$ auction was acquired when client $i$ had strictly less than $u_i$ shares. Consider the time $t$ at which $s$ was required, only now consider the $Q^T$ auction. There are only two possibilities: if $i$ holds $\leq u_i$ shares at time $t$, then $i$ will also acquire share $s$. In either case, client $i$ in the $Q^T$ auction holds at least as many shares as client $i$ in the $Q_i$ auction after time $t$, and so this holds at $t = T^*$, which proves the claim.

So we have shown that, given a fixed price trajectory $\vec{p}$, if client $i$ is satisfied at termination, $v^i_A(Q^T|\vec{p}) \geq v^i_A(Q_i|\vec{p})$.

We claim that for at least $(1-\sqrt{\beta + (1-\beta)\rho})n$ clients $i$, the probability that $i$ is unsatisfied at termination is less than $\sqrt{\beta + (1-\beta)\rho}$. This is clear because we know that if the error of the private counter does not exceed $E$, which happens with probability $(1-\beta)$, the number of unsatisfied clients at termination is less than $\rho n$. Since the number of unsatisfied bidders cannot exceed $n$, this implies the expected number of unsatisfied bidders is bounded by $(\beta + (1-\beta)\rho)n$. But, if for more than $\sqrt{\beta + (1-\beta)\rho}$ clients, the probability of being unsatisfied exceeded $\sqrt{\beta + (1-\beta)\rho}$, the expected number of unsatisfied bidders would exceed $(\beta + (1-\beta)\rho)n$, which is a contradiction. \\

Let $i$ be one such client, and let $I$ be the indicator that $i$ is satisfied at termination. Now:

$$v^i_A(Q^T) = \mathbb{E}_{(I, \vec{p}) \sim \pauction(Q^T)}[v^i_A(Q^T)|\vec{p}, I] \geq e^{-\epsilon}\mathbb{E}_{(I, \vec{p}) \sim \pauction(Q_i)}[v^i_A(Q_i) |\vec{p}, I] = $$
 $$e^{-\epsilon}\int_{\vec{p}} v^i_A(Q_i | \vec{p} , I =1)\Pr[\vec{p}, I = 1] +  v^i_A(Q_i|\vec{p}, I = 0)\Pr[I = 0 , p]  \geq $$
 $$e^{-\epsilon}\int_{\vec{p}} v^i_A(Q_i | \vec{p} , I =1)\Pr[\vec{p}, I = 1]  = e^{-\epsilon}\Pr[I = 1]v^i_A(Q_i|I = 1)$$

Finally, since $\Pr[I = 1] \geq 1-\sqrt{\beta + (1-\beta)\rho}, \Pr[I = 1]\Pr[I = 1]v^i_A(Q_i|I = 1) \geq e^{-\epsilon}\Pr[I = 1]v^i_A(Q_i)-e^{-\epsilon}\frac{\sqrt{\beta + (1-\beta)\rho}}{1-\sqrt{\beta + (1-\beta)\rho}}$.
\end{proof}

\section{Proofs from Section~\ref{sec:reduction}.}
\begin{algorithm}
\caption{Greedy Mechanism}
\begin{algorithmic}[1]
\label{mech:greedy}
\Procedure{$\A$}{Utility distributions $U_i$ for $n$ clients, $V$ shares to allocate at each of $T$ rounds}
\For{$t = 1\ldots T$}
       \For{i = 1 \ldots n}
		\State Client $i$ draws $u_{it} \sim U_i$
		\State Client $i$ picks request distribution $Q_{it} = L_t^i(\mathcal{H}_{t}^{i}, u_{it})$
		\State Client $i$ draws $r_{it} \sim Q_{it}$, and submits $r_{it}$
	\EndFor
	\State $\A$ updates its estimates $\hat{Q}_i(r_{it}) = \textbf{1}_{r_{it}}$ 
	\State $\A$ computes an allocation $S_t = A(\hat{Q}_1(r_{1t}), \ldots \hat{Q}_t(r_{nt}))$
	\State $\A$ observes the executed shares $v_i(S_t)$ for each client
	\State $\A$ updates its estimates of the conditionals $\hat{Q}_i(r_{it})$
	\State $\A$ updates the history: $H_{t+1} = H_t \cup (r_{it}, s_{it}, v_i(S_t))_{i=1}^n$
\EndFor
\EndProcedure
\end{algorithmic}
\end{algorithm}

\begin{algorithm}
\caption{Greedy Private Mechanism}
\begin{algorithmic}[1]
\label{mech:priv}
\Procedure{$\A$}{Utility distributions $U_i \in \Delta([U])$ for $n$ clients, $V$ shares to allocate at each of $T$ rounds, $\pauction, \epsilon, \alpha$}
\For{$t = 1\ldots T$}
       \For{i = 1 \ldots n}
		\State Client $i$ draws $u_{it} \sim U_i$
		\State Client $i$ picks request distribution $Q_{it} = L_t^i(\mathcal{H}_{t}^{i}, u_{it})$ 
		\State Client $i$ draws $r_{it} \sim Q_{it}$, and submits $r_{it}$
	\EndFor
	\State $\A$ updates its estimates $\hat{Q}_i(r_{it}) = \textbf{1}_{r_{it}}$ 
	\State $\A$ computes an allocation $S_t = \pauction(\hat{Q}_1(r_{1t}), \ldots \hat{Q}_t(r_{nt}), \epsilon, \alpha)$
	\State $\A$ observes the executed shares $v_i(S_t)$ for each client
	\State $\A$ updates its estimates of the conditionals $\hat{Q}_i(r_{it})$
	\State $\A$ updates the history: $H_{t+1} = H_t \cup (r_{it}, s_{it}, v_i(S_t))_{i=1}^n$
\EndFor
\EndProcedure
\end{algorithmic}
\end{algorithm}

\textbf{Proof of Theorem~\ref{thm:apxtrue}.}
\begin{proof}
The crux of the proof relies on the fact that once we have fixed the request distributions $Q_{-it}$ at each round of all other players, 
and given that Algorithm~\ref{mech:priv} assumes the clients are truthful, requesting truthfully is an approximately dominant strategy for most of the
clients by Theorem~\ref{thm:privtrue}. Formally, by Theorem~\ref{thm:privtrue}, fixing $Q_{-it}$,  in any fixed round the truthfulness guarantee applies to a $1-\beta^2/T$
fraction of the clients $i$. Hence over all $T$ rounds it applies to a $1-\beta^2$ fraction of the clients. This implies that over the random draw of 
$Q_{-it}$, for at least $1-\beta$ fraction of clients, the probability of the truthfulness guarantee holding in every round is $\geq 1-\beta$. Else, 
the expectation of the total number of unsatisfied clients over all $T$ rounds would be strictly greater than $\beta n \cdot \beta = \beta^2 n$, which contradicts the fact that with probability $1$ it is $\leq \beta^2 n$ by Theorem~\ref{thm:privtrue}. Let $i$ be one such client. 

We will also require the fact that under $L_{i}^{t}$ or under $L_{i*}^{t}$ any given realization of $Q_{-it}$ is equally likely. We first observe that the outputs of the mechanism to each client are only a function of the estimated distributions $\hat{Q}_i(r_{it})$ which is the only dependence the allocations have on the strategies. Hence privacy in the $\hat{Q}_i(r_{it})$ guarantees privacy in the $L_{i}^{t}$. In the case where $\hat{Q}_i(r_{it}) = \mathbf{1}_{r_{it}}$, the overall procedure is $(\epsilon, \beta)$-JDP, whereas if the estimation procedure uses data from all of the rounds we have to use the composition rule for differential privacy to get an overall $\tilde{O}(\sqrt{T}\epsilon, \beta T)$ privacy guarantee. Then since the overall mechanism is $(\epsilon, \beta)$-JDP in the strategies $L_{i}^{t}$, 
the client distributions at each round $Q_{-it}$, which are a post-processing of the outputs of the mechanism to all of the other clients, are an $(\epsilon, \beta)$-joint differentially private function of $L_i^{t}$. We are now equipped to show the main result:

$$
v_i(L_{i}^{1}, \ldots, L_{i}^{n}) = \int_{Q_{-it} \in \Delta([U])^{(n-1)T}}\sum_{t=1}^{T}\E[v_i(S_t)|Q_{-it}, (L_{i}^{t})]\Pr[Q_{-it}|(L_{i}^{t})]
$$
 By the argument above, we also know that there exists a  subset $\Omega \subset \Delta([U])^{(n-1)T}$ such that $\Pr[Q_{-it} \in \Omega] \geq 1-\beta$ and for all $Q_{-it} \in \Omega$, client $i$ is approximately truthful. Hence:
$$
v_i(L_{i}^{1}, \ldots, L_{i}^{n}) \leq \int_{Q_{-it} \in \Omega}\sum_{t=1}^{T}\E[v_i(S_t)|Q_{-it}, (L_{i}^{t})]\Pr[Q_{-it}|(L_{i}^{t})] + \beta UT
$$
Moreover, by Theorem~\ref{thm:privtrue} for $Q_{-it} \in \Omega, \E[v_i(S_t)|Q_{-it}, (L_{i}^{t})] \leq e^{\epsilon}\E[v_i(S_t)|Q_{-it}, (L_{i*}^{t})] + \frac{\beta^2/T}{1-\beta^2/T}$, and by  $(\epsilon, \beta)$-JDP $\Pr[Q_{-it}|(L_{i}^{t})] \leq e^{\epsilon}\Pr[Q_{-it}|(L_{i*}^{t})] + \beta$. Substituting both of these inequalities into the above equation gives:
$$
v_i(L_{i}^{1}, \ldots, L_{i}^{n}) \leq e^{2\epsilon}\int_{Q_{-it} \in \Omega}\sum_{t=1}^{T}\E[v_i(S_t)|Q_{-it}, (L_{i*}^{t})]\Pr[Q_{-it}|(L_{i*}^{t})] + 2\beta UT + e^{\epsilon}\frac{\beta^2}{1-\beta^2/T},
$$
Giving:
$$
v_i(L_{i}^{1}, \ldots, L_{i}^{n}) \leq e^{2\epsilon}v_i(L_{i}^{1}, \ldots, L_{i}^{n}) + 2\beta UT + e^{\epsilon}\frac{\beta^2}{1-\beta^2/T},
$$
as desired.
\end{proof}

\textbf{Sketch of the Naive Learning Mechanism.}\\
\begin{enumerate}
\item Taking $\tau \approx \frac{1}{n U}$, after $\approx T/\tau = \text{poly}(n, U, T)$ rounds, for any $r \in [U], i \in [n]$ such that $Q_i(r_i) > \tau$,  we will have observed $\approx T$ draws from the conditional distribution $Q_i(u_i|r_i)$.
\item Hence, after a polynomial number of rounds, we can learn each conditional $Q_i(u_i|r_i)$ arbitrarily well for any $i, r_i$ such that $Q_i(r_i) > \tau$.
\item The lender uses the observed samples to compute estimates $\hat{Q}_i(u|r_i)$ for every $i, r_i$, and upon observing a draw $\vec{r} = (r_1, r_2, \ldots r_n) \sim Q$, uses \emph{Greedy} to compute the optimal allocation with respect to these $\hat{Q}_i$.
\item A union bound shows that with high probability for every $r_i = \vec{r}_i$, $Q_i(r_i) > \tau$, and hence we have a good estimate of each of the conditional distributions.
\item Standard Chernoff bounds characterize the sample complexity of learning and show that the optimal allocation with respect to $\hat{Q}_i$ is approximately optimal with respect to the true distributions $Q_i$.
\end{enumerate}

Note that we may view the ``learning'' phase as a \emph{means} by which the clients can communicate distributions $\hat{Q_i}$ to the lender. We note that (so long as the lender observes each report $r_i$ at least once in the learning phase), \emph{if} client $i$ is reporting truthfully, we will have $\hat{Q_i} = Q_i^T$, the truthful reporting distribution. Since it is a dominant strategy to report $Q^T$ even when the client has the ability to report \emph{any} distribution, since $Q_i^T$ is possible for the client to report through this more restricted learning interface, truthful reporting remains a dominant strategy.

\begin{example}
\label{example}
Consider the case of $V = 1, n = 2$ clients over $T = 100$ rounds, where both clients have demands fixed at $1$ share in every round. Further, suppose client $1$ has the following strategy: if client $1$ receives a share in the first round, she will not request any more shares for the subsequent $99$ rounds. If she does not receive a share, she will play truthfully for the rest of the rounds, requesting $1$ share at each round. In this game the dominant strategy for client $2$ is to request $0$ shares in the first round and then be guaranteed a payoff of $99$ shares over the subsequent rounds, as opposed to playing truthfully at all rounds, which if ties are broken randomly gets her $50$ shares in expectation. 
\end{example}

\end{document}